\documentclass[11pt]{amsart}
\usepackage{graphicx, amssymb}

\newtheorem{thm}{Theorem}[section]

\newtheorem{prop}[thm]{Proposition}
\theoremstyle{definition}
\newtheorem{defn}[thm]{Definition}

\theoremstyle{remark}
\newtheorem{rem}[thm]{Remark}

\numberwithin{equation}{section}

\newcommand{\set}[1]{\left\{#1\right\}}
\newcommand{\Real}{\mathbb R}
\newcommand{\Natural}{\mathbb N}

\newcommand{\such}{\, | \,}
\newcommand{\limn}{\lim_{n \to \infty}}
\newcommand{\limk}{\lim_{k \to \infty}}
\newcommand{\dfn}{\, := \,}

\newcommand{\prob}{\mathbb{P}}

\newcommand{\plim}{\Lb^0 \text{-} \lim}
\newcommand{\plimn}{\plim_{n \to \infty}}

\newcommand{\qprob}{\mathsf{Q}}
\newcommand{\qprobb}{\mathbb{Q}}
\newcommand{\expec}{\mathbb{E}}

\newcommand{\Lb}{\mathbb{L}}
\newcommand{\F}{\mathcal{F}}

\newcommand{\cadlag}{c\`adl\`ag}

\newcommand{\ud}{\, \mathrm d}

\newcommand{\inner}[2]{\left \langle #1 , \, #2 \right \rangle}

\newcommand{\num}{num\'eraire}
\newcommand{\X}{\mathcal{X}}
\newcommand{\Xxx}{\X^{\times \times} (1)}
\newcommand{\V}{\mathcal{V}^{\downarrow} (1)}
\newcommand{\Y}{\mathcal{Y}}

\newcommand{\trace}{\mathsf{trace}}

\newcommand{\tX}{\widetilde{X}}

\newcommand{\pare}[1]{\left(#1\right)}
\newcommand{\bra}[1]{\left[#1\right]}
\newcommand{\dbra}[1]{[\kern-0.15em[ #1 ]\kern-0.15em]}
\newcommand{\dbraco}[1]{[\kern-0.15em[ #1 [\kern-0.15em[}

\newcommand{\indic}{\mathbb{I}}

\newcommand{\NFLVR}{\emph{NFLVR}}

\newcommand{\ELMM}{\emph{ELMM}}
\newcommand{\WELMM}{\emph{WELMM}}

\newcommand{\NAone}{\emph{NA$_1$}}

\newcommand{\absco}{{<\kern-0.53em<}}

\newcommand{\oxi}{\overline{\xi}}
\newcommand{\oY}{\overline{Y}}

\newcommand{\Dya}{\mathbb D}

\begin{document}

\title[Finitely additive probabilities and the FTAP]{Finitely additive probabilities and the Fundamental Theorem of Asset Pricing}%
\author{Constantinos Kardaras}%
\address{Constantinos Kardaras, Mathematics and Statistics Department, Boston University, 111 Cummington Street, Boston, MA 02215, USA.}%
\email{kardaras@bu.edu}%

\thanks{This work was partially supported by the National
Science Foundation, grant DMS-0908461}%
\date{\today}%
\dedicatory{Dedicated to Prof. Eckhard Platen, on the occasion of his 60$^\text{th}$ birthday.}%
\begin{abstract}
This work aims at a deeper understanding of the mathematical implications of the economically-sound condition of \emph{absence of arbitrages of the first kind} in a financial market. In the spirit of the Fundamental Theorem of Asset Pricing (FTAP), it is shown here that absence of arbitrages of the first kind in the market is equivalent to the existence of a finitely additive probability, weakly equivalent to the original and only locally countably additive, under which the discounted wealth processes become ``local martingales''. The aforementioned result is then used to obtain an independent proof of the classical FTAP, as it appears in \cite{MR1304434}. Finally, an elementary and short treatment of the previous discussion is presented for the case of continuous-path semimartingale asset-price processes.
\end{abstract}

\maketitle

\setcounter{section}{-1}

\section{Introduction}

In the Quantitative Finance literature, the most common normative assumption placed on financial market models in the literature is the existence of an Equivalent Local Martingale Measure (ELMM), i.e., a probability, equivalent to the original one, that makes discounted asset-price processes local martingales. There is, of course, a very good reason for postulating the existence of an ELMM in the market: the Fundamental Theorem of Asset Pricing (FTAP) establishes\footnote{At least in the case where asset-price processes are nonnegative semimartingales; see \cite{MR1671792} for the case of general semimartingales, where $\sigma$-martingales, a generalization of local martingales, have to be utilized.} the equivalence between a precise market viability condition, coined ``No Free Lunch with Vanishing Risk'' (NFLVR) with the existence of an ELMM (see \cite{MR1304434} and \cite{MR1671792}).

The importance of condition NFLVR notwithstanding, there has lately been considerable interest in researching models where an ELMM might fail to exist. Major examples include the benchmark approach in financial modeling of \cite{MR2267213}, as well as the emergence of stochastic portfolio theory (\cite{FerKar07}), a \emph{descriptive} theory of financial markets. Even though the previous approaches allow for the existence of some form of arbitrage, they still deal with viable models of financial markets. In fact, the markets there satisfy a weaker version of the NFLVR condition; more precisely, there is \emph{absence of arbitrages of the first kind}\footnote{The terminology ``arbitrage of the first kind'' was introduced in \cite{Ing87}, although our definition is closer in spirit to arbitrages of the first kind in the context of large financial markets, as appears in \cite{MR1348197}. One should also mention \cite{MR1774056}, where arbitrages of the first kind are called \textsl{cheap thrills}.} (see Definition \ref{dfn: arb first kind} of the present paper), which we abbreviate as condition NA$_1$. In the recent work \cite{Kar09a}, it was shown that condition NA$_1$ is equivalent to the existence of a \emph{strictly positive local martingale deflator}, i.e., a strictly positive process with the property that every asset-price, when deflated by it, becomes a local martingale. The previous mathematical counterpart of the economic NA$_1$ condition is rather elegant; however, and in order to provide a closer comparison with the FTAP of \cite{MR1304434}, it is still natural to wish to equivalently express the NA$_1$ condition in terms of the existence of some measure that makes discounted asset-prices have some kind of martingale property.

\smallskip

In an effort to connect, expand, and simplify previous research, the purpose of this paper is threefold; in particular, we aim at:
\begin{enumerate}
	\item presenting a weak version of the FTAP, stating the equivalence of the NA$_1$ condition with the existence of a ``probability'' that makes discounted nonnegative wealth processes ``local martingales'';
	\item using the previous result as an intermediate step to obtain the FTAP as it appears in \cite{MR1304434};
 	\item providing an elementary proof of the above weak version of the FTAP discussed in (1) above when the asset-prices are continuous-path semimartingales.
\end{enumerate}

In order to tackle (1), we introduce the concept of a Weakly Equivalent Local Martingale Measure (WELMM). A WELMM is a \emph{finitely additive} probability\footnote{Finitely additive measures have appeared quite often in in economic theory in a financial equilibrium setting in cases of infinite horizon (see \cite{RePEc:ier:iecrev:v:33:y:1992:i:2:p:323-39})
or even finite-time horizon with credit constraints on economic agents (see \cite{MR1774056} and \cite{MR1748373}).} that is \emph{locally countably additive} and makes discounted asset-price processes behave like local martingales. Of course, the last local martingale property has to be carefully and rigorously defined, as only finitely additive probabilities are involved --- see Definition \ref{dfn: local marts} later on in the text. In Theorem \ref{thm: main}, and in a general semimartingale market model, we obtain the equivalence between condition NA$_1$ and the existence of a WELMM.

Theorem \ref{thm: main} can be also seen as an intermediate step in proving the FTAP of \cite{MR1304434}. Under the validity of Theorem \ref{thm: main}, and using the very important optional decomposition theorem, this task becomes easier, as the proof of Theorem \ref{thm closedness in L0} of the present paper shows.

We now come to the issue raised at (3) above. In order to establish our weak version of the FTAP, we need to invoke the main result from \cite{Kar09a}, which itself depends heavily upon results of \cite{MR2335830}. The immense level of technicality in the proofs of the previous results render their presentation in graduate courses almost impossible. The same is true for the FTAP of \cite{MR1304434}. Given the importance of such type of results, this is really discouraging. We provide here a partial resolution to this issue in the special case where the asset-prices are continuous-path semimartingales. As is shown in Theorem \ref{thm: FTAP_cont}, proving of our main Theorem \ref{thm: main} becomes significantly easier; in fact, the only non-trivial result that is used in the course of the proof is the representation of a continuous-path local martingales as time-changed Brownian motion. Furthermore, in Theorem \ref{thm: FTAP_cont}, condition NA$_1$ is shown to be equivalent to the existence and square-integrability of a risk-premium process, which has nice economic interpretation and can be easily checked once the model is specified.

\bigskip

The structure of the paper is as follows. In Section \ref{sec: weak version of FTAP}, the market is introduced, arbitrages of the first kind and the concept of a WELMM are defined, and Theorem \ref{thm: main}, the weak version of the FTAP, is stated. Section \ref{sec: the FTAP of K-D-S} deals with a proof of the FTAP as it appears in \cite{MR1304434}. Finally, Section \ref{sec: cont semimarts} contains the statement and elementary proof of Theorem \ref{thm: FTAP_cont}, which is a special case of Theorem \ref{thm: main} when the asset-price processes are continuous-path semimartingales.

\section{Arbitrages of the First Kind and Weakly Equivalent Local Martingale Measures} \label{sec: weak version of FTAP}

\subsection{General probabilistic remarks}
All stochastic processes in the sequel are defined on a filtered probability space $\left(\Omega, \, \F, \, (\F_t)_{t \in \Real_+}, \, \prob\right)$. Here, $\prob$ is a probability on $(\Omega, \F)$, where $\F$ is a $\sigma$-algebra that will make all involved random variables measurable. The filtration $(\F_t)_{t \in \Real_+}$ is assumed to satisfy the usual hypotheses of right-continuity and saturation by $\prob$-null sets. A finite financial planning horizon $T$ will be assumed. Here, $T$ is a $\prob$-a.s. finite stopping time and all processes will be assumed to be constant, and equal to their value they have at $T$, after time $T$. It will be assumed throughout that $\F_0$ is trivial modulo $\prob$ and that $\F_T = \F$.

\subsection{The market and investing}

Henceforth, $S$ will be denoting the \emph{discounted}, with respect to some baseline security, price process of a financial asset, satisfying:
\begin{equation} \label{S-MART} \tag{S-MART}
S \text{ is a nonnegative semimartingale.}
\end{equation}

Starting with capital $x \in \Real_+$, and investing according to some predictable and $S$-integrable strategy $\vartheta$, an economic agent's discounted wealth is given by the process
\begin{equation} \label{eq: wealth process, all}
X^{x, \vartheta} \dfn x + \int_0^\cdot \vartheta_t \ud S_t.
\end{equation}
In frictionless, continuous-time trading, credit constraints have to be imposed on investment in order to avoid doubling strategies. Define then $\X(x)$ to be the set of all wealth processes $X^{x, \vartheta}$ in the notation of \eqref{eq: wealth process, all} such that $X^{x, \vartheta} \geq 0$. Also, let $\X \dfn \bigcup_{x \in \Real_+} \X(x)$ denote the set of all nonnegative wealth processes.

\subsection{Arbitrages of the first kind}

The market viability notion that will be introduced now will be of central importance in our discussion.
\begin{defn} \label{dfn: arb first kind}
An $\F_T$-measurable random variable $\xi$ will be called an \textsl{arbitrage of the first kind} if $\prob[\xi \geq 0] = 1$, $\prob[\xi > 0] > 0$, and \emph{for all $x > 0$ there exists $X \in \X(x)$, which may depend on $x$, such that $\prob[X_T \geq \xi] = 1$}.

If there are no arbitrages of the first kind in the market, we say that condition NA$_1$ holds.
\end{defn}

In view of Proposition 3.6 from \cite{MR1304434}, condition NA$_1$ is weaker than condition NFLVR. In fact, condition NA$_1$ is exactly the same as condition ``No Unbounded Profit with Bounded Risk'' (NUPBR) of \cite{MR2335830}, as we now show.

\begin{prop} \label{prop: NA_one iff NUPBR}
Condition \NAone \ is equivalent to the requirement that the set $\set{X_T \such X \in \X(1)}$ is bounded in probability.
\end{prop}

\begin{proof}
Using the fact that $\X(x) = x \X(1)$ for all $x > 0$, it is straightforward to check that if an arbitrage of the first kind exists, then $\set{X_T \such X \in \X(1)}$ is not bounded in probability. Conversely, assume that $\set{X_T \such X \in \X(1)}$ is not bounded in probability. Since $\set{X_T \such X \in \X(1)}$ is further convex, Lemma 2.3 of \cite{MR1768009} implies the existence of $\Omega_u \in \F_T$ with $\prob[\Omega_u] > 0$ such that, for all $n \in \Natural$, there exists $\tX^n \in \X(1)$ with $\prob[\{\tX^n_T \leq  n \} \cap \Omega_u] \leq \prob[\Omega_u] / 2^{n+1}$. For all $n \in \Natural$, let $A^n = \indic_{\{\tX^n_T > n\}} \cap \Omega_u \in \F_T$. Then, set $A := \bigcap_{n \in \Natural} A^n \in \F_T$ and $\xi :=  \indic_A$. It is clear that $\xi$ is $\F_T$-measurable and that $\prob[\xi \geq 0] = 1$. Furthermore, since $A \subseteq \Omega_u$ and
\[
\prob \bra{\Omega_u \setminus A} = \prob \bra{\bigcup_{n \in \Natural} \pare{\Omega_u \setminus A^n}} \leq \sum_{n \in \Natural} \prob \bra{ \Omega_u \setminus A^n} = \sum_{n \in \Natural} \prob \bra{ \{ \tX^n_T \leq  n \} \cap \Omega_u} \leq \sum_{n \in \Natural} \frac{\prob[\Omega_u]}{2^{n+1}} = \frac{\prob[\Omega_u]}{2},
\]
we obtain $\prob[A] > 0$, i.e., $\prob[\xi > 0] > 0$. For all $n \in \Natural$ set $X^n := (1 / n) \tX^n$, and observe that $X^n \in \X(1 / n)$ and $\xi = \indic_{A} \leq \indic_{A^n} \leq X^n_T$ hold for all $n \in \Natural$. It follows that $\xi$ is and arbitrage of the first kind, which finishes the proof.
\end{proof}

%
%

\subsection{Weakly Equivalent Local Martingale Measures} \label{sec: WELMM}

The mathematical counterpart of the NA$_1$ condition involves a weakening of the concept of an ELMM. The appropriate notion involves measures that behave like probabilities, but are \emph{finitely additive} and only \emph{locally countably additive}.

In what follows, a \textsl{localizing sequence} will refer to a \emph{nondecreasing} sequence $(\tau^n)_{n \in \Natural}$ of stopping times such that $\uparrow \limn \prob [\tau^n \geq T] = 1$.

\subsubsection{Local probabilities weakly equivalent to $\prob$}

The concept that will be introduced below in Definition \ref{defn: local prob, weak equiv} is essentially a localization of countably additive probabilities.

\begin{defn} \label{defn: local prob, weak equiv}
A mapping $\qprob : \F \mapsto [0,1]$ is a \textsl{local probability weakly equivalent to $\prob$} if:
\begin{enumerate}
  \item $\qprob[\emptyset] = 0$, $\qprob[\Omega] = 1$, and $\qprob$ is (finitely) additive: $\qprob[A \cup B] = \qprob[A] + \qprob[B]$ whenever $A \in \F$ and $B \in \F$ satisfy $A \cap B = \emptyset$;
  \item for $A \in \F$, $\prob[A] = 0$ implies $\qprob[A] = 0$;
  \item there exists a localizing sequence $(\tau_n)_{n \in \Natural}$ such that, when restricted on $\F_{\tau_n}$, $\qprob$ is countably additive and equivalent to $\prob$, for all $n \in \Natural$. (Such sequence of stopping times will be called a \textsl{localizing sequence for $\qprob$}.)
\end{enumerate}
\end{defn}

Conditions (1) and (2) above imply that $\qprob$ is a positive element of the dual of $\Lb^\infty$, the space of (equivalence classes modulo $\prob$ of) $\F$-measurable random variable that are bounded modulo $\prob$ equipped with the essential-sup norm. The theory of finitely additive measures is developed in great detail in \cite{MR751777}; for our purposes here, mostly results from the Appendix of \cite{MR1841719}, as well as some material from \cite{MR2016601}, will be needed.

To facilitate the understanding, finitely additive positive measures that are not necessarily countably additive will be denoted using sans-serif typeface (like ``$\qprob$''), while for countably additive probabilities the blackboard bold typeface (like ``$\qprobb$'') will be used. As $\qprob$ will be in the dual of $\Lb^\infty$, $\inner{\qprob}{\xi}$ will denote the action of $\qprob$ on $\xi \in \Lb^\infty$.
The fact that $\qprob$ is a positive functional enables to extend the definition of $\inner{\qprob}{\xi}$ for $\xi \in \Lb^0$ with $\prob[\xi \geq 0] = 1$, via $\inner{\qprob}{\xi} \dfn \limn \inner{\qprob}{\xi \indic_{\set{\xi \leq n}}} \in [0, \infty]$. ($\Lb^0$ denotes the set of all $\prob$-a.s. finitely-valued random variables modulo $\prob$-equivalence equipped with the topology of convergence in probability.)

\begin{rem} \label{rem: weak vs strong}
In general, a finitely additive probability $\qprob: \F \mapsto [0,1]$ is called \textsl{weakly absolutely continuous with respect to $\prob$} if for each $A \in \F$ with $\prob[A] = 0$ we have $\qprob[A] = 0$. Furthermore, $\qprob$ is called \textsl{strongly absolutely continuous with respect to $\prob$} if for any $\epsilon > 0$ there exists $\delta = \delta(\epsilon) > 0$ such that $E \in \F$ and $\prob[E] < \delta$ implies $\qprob[E] < \epsilon$. It is clear that strong absolute continuity of $\qprob$ with respect to $\prob$ is a stronger requirement than weak absolutely continuity of $\qprob$ with respect to $\prob$. Actually, the two notions coincide when $\qprob$ is countable additive. Of course, similar definitions can be made with the roles of $\prob$ and $\qprob$ reversed. Then, $\prob$ and $\qprob$ are called \textsl{weakly} (respectively, \textsl{strongly}) \textsl{equivalent} if $\qprob$ is weakly (respectively, strongly)  absolutely continuous with respect to $\prob$ and $\prob$ is weakly (respectively, strongly)  absolutely continuous with respect to $\qprob$.

In Definition \ref{defn: local prob, weak equiv}, $\qprob$ was called a local probability ``weakly equivalent to $\prob$''; however, condition (2) only implies that $\qprob$ is weakly absolutely continuous with respect to $\prob$. We claim that $\prob$ is also weakly absolutely continuous with respect to $\qprob$. Indeed, let $\qprob$ satisfy (1) and (3) of Definition \ref{defn: local prob, weak equiv}. Pick any $A \in \F$ with $\qprob[A] = 0$. Since $A \cap \{ \tau_n \geq T\} \in \F_{\tau_n}$ for all $n \in \Natural$, $\qprob[A \cap \{ \tau_n \geq T\}] = 0$ implies that $\prob[A \cap \{ \tau_n \geq T\}] = 0$ by (3). Then, $\prob[A] = \, \uparrow \limn \prob[A \cap \{ \tau_n \geq T\}] = 0$.

Let $\qprob$ be a local probability weakly equivalent to $\prob$. When $\qprob$ is only finitely, but not countably,  additive, $\prob$ and $\qprob$ are not strongly equivalent, as we now explain. Write $\qprob = \qprob^r + \qprob^s$ for the unique decomposition of $\qprob$ into its regular and singular part. (The regular part $\qprob^r$ is countably additive, while the singular part $\qprob^s$ is purely finitely additive, meaning that there is no nonzero countably additive measure that is dominated by $\qprob^s$. One can check \cite{MR751777} for more information.) According to Lemma A.1 in \cite{MR1841719}, for all $\epsilon > 0$ one can find a set $A_\epsilon \in \F$ with $\prob[A_\epsilon] < \epsilon$ and $\qprob^s[A_\epsilon] = \qprob^s[\Omega]$; therefore $\qprob[A_\epsilon] \geq \qprob^s[\Omega]$. In other words, if $\qprob^s$ is nontrivial, then $\qprob$ is \emph{not} strongly absolutely continuous with respect to $\prob$. Note, however, that $\prob$ \emph{is} strongly absolutely continuous with respect to $\qprob$ in view of condition (3) of Definition \ref{defn: local prob, weak equiv}.
\end{rem}

We briefly digress from our main topic to give a simple criterion that connects the countable additivity of $\qprob$, a local probability weakly equivalent to $\prob$, with the strong equivalence between $\qprob$ and $\prob$, as the latter notion was introduced in Remark \ref{rem: weak vs strong} above.

\begin{prop}
Let $\qprob$ be a local probability weakly equivalent to $\prob$. The following are equivalent:
\begin{enumerate}
  \item $\qprob$ is countably additive, i.e., a true probability.
  \item $\qprob$ is strongly absolutely continuous with respect to $\prob$. 
  \item $\uparrow \lim_{n \to \infty} \qprob[\tau^n \geq T] = 1$ holds for \emph{any} localizing sequence $(\tau^n)_{n \in \Natural}$  for $\qprob$.
  \item $\uparrow \lim_{n \to \infty} \qprob[\tau^n \geq T] = 1$ holds for \emph{some} localizing sequence $(\tau^n)_{n \in \Natural}$ for $\qprob$.
\end{enumerate}
\end{prop}

\begin{proof}
The implications (1) $\Rightarrow$ (2) $\Rightarrow$ (3) $\Rightarrow$ (4) are straightforward, so we only focus on the implication (4) $\Rightarrow$ (1). Let $(E^k)_{k \in \Natural}$ be a decreasing sequence of $\F$-measurable sets such that $\bigcap_{k \in \Natural} E^k = \emptyset$. We need show that $\downarrow \limk \qprob[E^k] = 0$. Consider the $\qprob$-localizing sequence $(\tau^n)_{n \in \Natural}$ of statement (4). For each $n \in \Natural$ and $k \in \Natural$ we have $E^k \cap \set{\tau^n \geq T} \in \F_{\tau^n}$. (Here, remember that $\F = \F_T$). This means that $\limsup_{k \to \infty} \qprob[E^k] \leq \qprob[\tau^n < T] + \limsup_{k \to \infty} \qprob[E^k \cap \set{\tau^n \geq T}] = \qprob[\tau^n < T]$, the last equality holding because $\qprob$ is countably additive on $\F_{\tau^n}$, for all $n \in \Natural$. Sending $n$ to infinity and using (4), we obtain the result.
\end{proof}

\subsubsection{Density processes} \label{subsec: density processes}

For a local probability weakly equivalent to $\prob$ as in Definition \ref{defn: local prob, weak equiv}, one can associate a strictly positive local $\prob$-martingale $Y^\qprob$, as will be now described. For all $n \in \Natural$, consider the $\prob$-martingale $Y^{\qprob, \, n}$ defined by setting
\[
Y^{\qprob, \, n}_\infty \equiv Y^{\qprob, \, n}_T \dfn \frac{\ud \pare{\qprob|_{\F_{\tau_n}}}}{ \ud \pare{\prob|_{\F_{\tau_n}}}}.
\]
It is clear that, $\prob$-a.s., $Y^{\qprob, \, n}_0 = 1$ and $Y^{\qprob, \, n}_T > 0$. Furthermore, for all $n \in \Natural
\setminus \set{0}$, $Y^{\qprob, \, n} = Y^{\qprob, \, n - 1}$ on the stochastic interval $\dbra{0, \tau_{n-1}}$. Therefore, patching the processes $(Y^{\qprob, \, n})_{n \in \Natural}$ together, one can define a local $\prob$-martingale $Y^{\qprob}$ such that, $\prob$-a.s., $Y^{\qprob}_0 = 1$ and $Y^{\qprob}_T > 0$.

\begin{rem}
A general result in \cite{MR2016601} shows that a \emph{supermartingale} $Y^\qprob$ can be associated to a finitely additive measure $\qprob$ that satisfies (1) and (2) of Definition \ref{defn: local prob, weak equiv}, but not necessarily (3). The construction of $Y^\qprob$ in \cite{MR2016601} is messier than the one provided above, exactly because condition (3) of Definition \ref{defn: local prob, weak equiv} is not assumed to hold. In the special case described here, the two constructions coincide.
\end{rem}

A partial converse of the above construction is also possible. To wit, start with some local $\prob$-martingale $Y$ such that, $\prob$-a.s., $Y_0 = 1$ and $Y_T > 0$. If $(\tau_n)_{n \in \Natural}$ is a localizing sequence for $Y$, one can define for each $n \in \Natural$ a probability $\qprobb^n$, equivalent to $\prob$ on $\F$, via the recipe $\ud \qprobb^n \dfn Y_{\tau_n} \ud \prob $. By Alaoglu's Theorem (see, for example, Theorem 6.25, page 250 of \cite{MR1717083}), the sequence $(\qprobb^n)_{n \in \Natural}$ has some cluster point $\qprob$ for the weak* topology on the dual of $\Lb^\infty$, which will be a finitely-additive probability. Proposition A.1 of \cite{MR1841719} gives that $\ud \qprob^r / \ud \prob = Y_T$. It is easy to see that $\qprob$ is a local probability weakly equivalent to $\prob$, as well as that $Y^{\qprob} = Y$. (Note that, again by Proposition A.1 of \cite{MR1841719}, the sequence $(\qprobb^n)_{n \in \Natural}$ might have several cluster points, but all will have the same regular part. Therefore, $\qprob$ is not uniquely defined, but it is always the case that $Y^{\qprob} = Y$.)

\subsubsection{Local martingales} When $\qprob$ is a local probability weakly equivalent to $\prob$ and fails to be countably additive, the concept of a $\qprob$-martingale, and therefore also of a local $\qprob$-martingale, is tricky to state. The reason is that existence of conditional expectations requires $\qprob$ to be countably additive in order to invoke the Radon-Nikod\'ym Theorem. To overcome this difficulty, we follow an alternative route. Let $\qprobb$ be a probability measure, equivalent to $\prob$. According to the optional sampling theorem (see, for example, \S1.3.C in \cite{MR1121940}), a \cadlag \ process $X$ is a local $\qprobb$-martingale if and only if there exists a localizing sequence $(\tau_n)_{n \in \Natural}$ such that $\inner{\qprobb}{X_{\tau^n \wedge \tau}}= X_0$ for all $n \in \Natural$ and all stopping times $\tau$. This characterization makes the following Definition \ref{dfn: local marts} plausible.

\begin{defn} \label{dfn: local marts}
Let $\qprob$ be a local probability weakly equivalent to $\prob$. A nonnegative \cadlag \ process   $X$ will be called a \textsl{local $\qprob$-martingale} if there exists a localizing sequence $(\tau_n)_{n \in \Natural}$ such that $\inner{\qprob}{X_{\tau^n \wedge \tau}}= X_0$ for all $n \in \Natural$ and all stopping times $\tau$.
\end{defn}

Now, a characterization of local $\qprob$-martingales in terms of density processes will be given. This extends the analogous result in the case where $\qprob$ is countably additive.

\begin{prop} \label{prop: Z Q-mart iff Y Z P-mart}
Let $\qprob$ be a local probability weakly equivalent to $\prob$ and let $Y^\qprob$ be defined as in \S \ref{subsec: density processes}. A nonnegative process $X$ is a local $\qprob$-martingale if and only if $Y^\qprob X$ is a local $\prob$-martingale.
\end{prop}

\begin{proof}
Start by assuming that $X$ is a local $\qprob$-martingale. Since $\inner{\qprob}{X_{\tau^n \wedge \tau}} = X_0$ for all $n \in \Natural$ and \emph{all} stopping times $\tau$, where $(\tau_n)_{n \in \Natural}$ is a localizing sequence, $(\tau_n)_{n \in \Natural}$ can be assumed to also localize $\qprob$. Then, since $X_{\tau^n \wedge \tau} \in \F_{\tau^n}$ for all $n \in \Natural$ and all stopping times $\tau$, and since $\qprobb^n \dfn \qprob |_{\F_{\tau^n}}$ is countably additive with $\ud \qprobb^n / (\ud \prob |_{\F_{\tau^n}}) = Y^\qprob_{\tau^n}$, it follows that
\[
Y^\qprob_0 X_0 = X_0 =  \inner{\qprob}{X_{\tau^n \wedge \tau}} = \expec[Y^\qprob_{\tau^n} X_{\tau^n \wedge \tau}] = \expec [ \expec[Y^\qprob_{\tau^n} \such \F_{\tau^n \wedge \tau}] X_{\tau^n \wedge \tau} ] = \expec[Y^\qprob_{\tau^n \wedge \tau} X_{\tau^n \wedge \tau}]
\]
for all $n \in \Natural$ and all stopping times $\tau$. This means that $Y^\qprob X$ is a local $\prob$-martingale.

Conversely, suppose that $Y^\qprob X$ is a local $\prob$-martingale. Let $(\tau_n)_{n \in \Natural}$ be a localizing sequence for both $Y^\qprob X$ and $\qprob$. Then, for all $n \in \Natural$ and all stopping times $\tau$,
\[
X_0 = Y^\qprob_0 X_0 = \expec[Y^\qprob_{\tau^n \wedge \tau} X_{\tau^n \wedge \tau}] = \expec [ \expec[Y^\qprob_{\tau^n} \such \F_{\tau^n \wedge \tau}] X_{\tau^n \wedge \tau} ] = \expec[Y^\qprob_{\tau^n} X_{\tau^n \wedge \tau}] = \inner{\qprob}{X_{\tau^n \wedge \tau}}.
\]
Therefore, $X$ is a local $\qprob$-martingale.
\end{proof}

\subsubsection{Weakly equivalent local martingale measures}

As will be shown in Theorem \ref{thm: main}, the following definition gives the mathematical counterpart of the market viability condition NA$_1$.

\begin{defn} \label{dfn: WELMM}
A \textsl{weakly equivalent local martingale measure} (WELMM) $\qprob$ is a local probability weakly equivalent to $\prob$ such that $S$ is a local $\qprob$-martingale.
\end{defn}

\begin{rem}[On the semimartingale property of $S$] \label{rem: on semimart property}
Under the assumption that $S$ is nonnegative, the existence of a WELMM \emph{enforces} the semimartingale property on $S$. Indeed, write $S = (1 / Y^\qprob)(Y^\qprob S)$, where $\qprob$ is a WELMM and $Y^\qprob$ is the density defined in \S\ref{subsec: density processes}. Since $Y^\qprob$ is a local $\prob$-martingale with $Y^\qprob_T > 0$, $\prob$-a.s., and $Y^\qprob S$ is also a local $\prob$-martingale, both $1 / Y^\qprob$ and $Y^\qprob S$ are semimartingales, which gives that $S$ is a semimartingale.

Semimartingales are essential in frictionless financial modeling. This has been made clear in Theorem 7.1 of \cite{MR1304434}, where it was shown that if $S$ is locally bounded and \emph{not} a semimartingale, condition NFLVR using only simple trading strategies fails. Furthermore, from the treatment in \cite{KarPla07} it follows that, if $S$ is nonnegative and \emph{not} a semimartingale, one can construct an arbitrage of the first kind, \emph{even} if one uses only \emph{no-short-sale} and \emph{simple} strategies.
\end{rem}

If $S$ satisfies \eqref{S-MART}, it is straightforward to check that a probability $\qprobb$ equivalent to $\prob$ is an ELMM if and only if each $X \in \X$ is a local $\qprob$-martingale.
The following result extends the last equivalence in the case of a WELMM.

\begin{prop} \label{prop: S is Q-mart iff X has Q-marts}
Let $\qprob$ be a local probability weakly equivalent to $\prob$. If $S$ satisfies \eqref{S-MART}, then $S$ is a local $\qprob$-martingale if and only if every process $X \in \X$ is a local $\qprob$-martingale.
\end{prop}

\begin{proof}
Start by assuming that $S$ is a local $\qprob$-martingale. For $x \in \Real_+$, let $X^{x, \vartheta}$ in the notation of \eqref{eq: wealth process, all} be a wealth process in $\X(x)$. A use of the integration-by-parts formula gives
\[
Y^\qprob X^{x, \vartheta} = x + \int_0^\cdot \pare{X^{x, \vartheta}_{t -} - \vartheta_t S_{t -} }\ud Y^\qprob_t + \int_0^\cdot \vartheta_{t} \ud (Y^\qprob S)_t
\]
It follows that $Y^\qprob X^{x, \vartheta}$ is a positive martingale transform under $\prob$, and therefore a local $\prob$-martingale by the Ansel-Stricker Theorem (see \cite{MR1277002}).

Now, assume that every process in $\X$ is a local $\qprob$-martingale. Since $S \in \X$, $S$ is a local $\qprob$-martingale.
\end{proof}

\begin{rem}
Let $\qprob$ be a local probability weakly equivalent to $\prob$.
Proposition \ref{prop: Z Q-mart iff Y Z P-mart} combined with  Proposition \ref{prop: S is Q-mart iff X has Q-marts} imply that  $\qprob$ is a WELMM if and only if $Y^\qprob X$ is a local $\prob$-martingale for all $X \in \X$. In other words, the process $Y^\qprob$ is a \textsl{strict martingale density} in the terminology of \cite{Schwe95} (see also \cite{MR1651229}).
\end{rem}

\subsection{The main result} After the preparation of the previous sections, it is possible to state Theorem \ref{thm: main} below, which can be seen as a weak version of the FTAP in \cite{MR1304434}.

\begin{thm} \label{thm: main}
Suppose that $S$ satisfies \eqref{S-MART}. Then, there are no arbitrages of the first kind in the market if and only if a weakly equivalent local martingale measure exists.
\end{thm}

\begin{proof}
By Theorem 1.1 in \cite{Kar09a}, condition NA$_1$ is equivalent to the existence of a nonnegative \cadlag \ process $Y$ with $Y_0 = 1$, $Y_T > 0$, and such that $Y X$ is a local $\prob$-martingale for all $X \in \X$. Then, using also the discussion in \S \ref{subsec: density processes} and Proposition \ref{prop: Z Q-mart iff Y Z P-mart}, NA$_1$ holds if and only if there exists a local probability $\qprob$, weakly equivalent to $\prob$, such that $X$ is a local $\qprob$-martingale for all $X \in \X$. Proposition \ref{prop: S is Q-mart iff X has Q-marts} gives that $\qprob$ is a WELMM, which completes the proof.
\end{proof}

\begin{rem} \label{rem: proving main using DS}
\emph{If} the statement of the FTAP of \cite{MR1671792} is assumed, one can provide a proof of Theorem \ref{thm: main} using the ``change of \num'' technique of \cite{MR1381678}; a similar approach has been taken up in \cite{MR2284490}. We opt here to prove Theorem \ref{thm: main} \emph{directly}, using the result of \cite{Kar09a} that is not relying on previous heavy results. Then, the classical FTAP itself becomes a corollary, as we shall see in Section \ref{sec: the FTAP of K-D-S} below. There is no claim that the path followed here is shorter or less arduous than the one taken up in \cite{MR1671792}, but certainly it has different focus.
\end{rem}

\begin{rem} \label{rem: weaken ass}
As can be seen from its proof, Theorem \ref{thm: main} still holds if the nonnegativity assumption on $S$ is removed, as long as we agree to reformulate the notion of a WELMM $\qprob$, asking that each $X \in \X$ is a local $\qprob$-martingale.

Furthermore, Theorem \ref{thm: main} also holds without the assumption that $S$ is one-dimensional. Indeed, in Remark \ref{rem: proving main using DS} above it was discussed that Theorem \ref{thm: main} can be seen as a consequence of the FTAP in \cite{MR1671792}, which does not require $S$ to be one-dimensional. Unfortunately, in \cite{Kar09a} the assumption that $S$ is one-dimensional is being made, mostly in order to avoid immense technical difficulties in the proof of Theorem 1.1 there, which is used to prove Theorem \ref{thm: main} above.
\end{rem}

\begin{rem}
Undoubtedly, the notion of a WELMM is more complicated than that of an ELMM. However, checking the existence of a WELMM is fundamentally easier than checking whether an ELMM exists for the market. Indeed, in view of Theorem \ref{thm: main}, existence of a WELMM is equivalent to the existence of the \num \ portfolio in the market. For checking the existence of the latter, there exists a necessary and sufficient criterion in terms of the predictable characteristics of the discounted asset-price process, as was shown in \cite{MR2335830}. The details are rather technical, but if the asset-price process has continuous paths the situation is very simple --- see Section \ref{sec: cont semimarts} later.
\end{rem}

\section{The FTAP of Delbaen and Schachermayer} \label{sec: the FTAP of K-D-S}

In this subsection, a proof of the FTAP as appears in \cite{MR1304434} is given using the already-developed tools. Also, the $\qprob$-supermartingale property of wealth processes in $\X$ when $\qprob$ is a WELMM is examined, and it is shown that the latter property holds only under the existence of an ELMM.
\subsection{Proving the FTAP}
In the notation of the present paper, the main technical difficulty for proving the FTAP in \cite{MR1304434} is showing that the set $\set{g \in \Lb^0 \such 0 \leq g \leq X_T \text{ for some } X \in \X(1)}$ is closed in probability under the NFLVR condition. This implies the weak* closedness of the set of bounded superhedgeable claims starting from zero capital and therefore allows for the use of the Kreps-Yan separation theorem (see \cite{MR611252} and \cite{MR580127}) in order to conclude the existence of a separating measure.

There is a way to establish the aforementioned closedness in probability using Theorem \ref{thm: main} and some additional well-known results. In fact, a seemingly stronger statement than the one in \cite{MR1304434} will now be stated and proved.

\begin{thm} \label{thm closedness in L0}
Under the assumption that no arbitrages of the first kind are present in the market, the set $\set{g \in \Lb^0 \such 0 \leq g \leq X_T \text{ for some } X \in \X(1)}$ is closed in probability.
\end{thm}

\begin{proof}
Define $\V$ to be the class of nonnegative, adapted, \cadlag, \emph{nonincreasing} processes with $V_0 \leq 1$. Then, set\footnote{The notation ``$\Xxx$'' is borrowed from \cite{MR1883202} since it is \emph{suggestive} of the fact that $\Xxx$ is the process-bipolar of $\X(1)$, as is defined in \cite{MR1883202}. Note, however, that it actually remains to show that $\X(1)$ is closed in probability to actually have that bipolar relationship.} $\Xxx \dfn \X(1) \, \V \, = \, \{X V  \such X \in \X(1) \text{ and } V \in \V \}$. The statement of the Theorem can be reformulated to say that the \emph{convex} set $\{\xi_T \such \xi \in \Xxx \}$ is closed in $\Lb^0$. Consider therefore a sequence $(\xi^n)_{n \in \Natural}$ such that $\plimn \xi^n_T = \zeta$. It will be shown below that there exists $\xi^\infty \in \Xxx$ such that $\xi^\infty_T = \zeta$.

In what follows in the proof, the concept of Fatou-convergence is used, which will now be recalled. Define $\Dya \dfn \set{ k / 2^m \such k \in \Natural, \, m \in \Natural}$ to be the set of dyadic rational numbers in $\Real_+$. A sequence $(Z^n)_{n \in \Natural}$   of nonnegative \cadlag \ processes \textsl{Fatou-converges} to $Z^\infty$ if
\[
Z^\infty_t = \limsup_{\Dya \ni s \downarrow t} \Big( \limsup_{n \to \infty} Z^n_s \Big) = \liminf_{\Dya \ni s \downarrow t} \Big( \liminf_{n \to \infty} Z^n_s \Big)
\]
holds $\prob$-a.s. for all $t \in \Real_+$. Note that, since all processes are assumed to be constant after time $T$, for any $t \geq T$ the above relationship simply reads $Z^\infty_T = \limn Z^n_T$, $\prob$-a.s.

From Theorem \ref{thm: main} and Proposition \ref{prop: Z Q-mart iff Y Z P-mart}, under absence of arbitrages of the first kind in the market, there exists some nonnegative process $\oY$ with $\oY_0 = 1$ and $\oY_T > 0$, $\prob$-a.s., such that $\oY X$ is a local $\prob$-martingale for all $X \in \X(1)$. Then, $\oY \xi$ is a nonnegative $\prob$-supermartingale for all $\xi \in \Xxx$. Since $(\oY \xi^n)_{n \in \Natural}$ is a sequence of nonnegative $\prob$-supermartingales with $\oY_0 \xi^n_0 \leq 1$, Lemma 5.2(1) of \cite{MR1469917} gives the existence of a sequence $(\oxi^n)_{n \in \Natural}$ such that $\oxi^n$ is a convex combination of $\xi^n, \xi^{n+1}, \ldots$ for each $n \in \Natural$ (and therefore $\oxi^n \in \Xxx$ for all $n \in \Natural$, since $\Xxx$ is convex), and such that $(\oY \oxi^n)_{n \in \Natural}$ Fatou-converges to some nonnegative $\prob$-supermartingale $Z$. Obviously, $Z_0 \leq 1$. Also, since $\plimn (\oY_T \xi_T^n) = \oY_T \zeta$, one gets $Z_T = \oY_T \zeta$. Define $\xi^\infty := Z / \oY$. Then, $(\oxi^n)_{n \in \Natural}$ Fatou-converges to $\xi^\infty$ and $\xi^\infty_T = \zeta$. The last line of business is to show that $\xi^\infty \in \Xxx$.

First of all, $\xi^\infty_0 \leq 1$ and $\xi^\infty$ is nonnegative. Let $\Y(1)$ be the class of all nonnegative process $Y$ with $Y_0 = 1$, $\prob$-a.s., such that $Y X$ is a $\prob$-supermartingale for all $X \in \X(1)$. Of course, for all $Y \in \Y(1)$ and all $\xi \in \Xxx$, $Y \xi$ is a $\prob$-supermartingale. It follows that $Y \oxi^n$ is a nonnegative $\prob$-supermartingale for all $n \in \Natural$. Since, for any $Y \in \Y(1)$, $(Y \oxi^n)_{n \in \Natural}$ Fatou-converges to $Y \xi^\infty$, using Fatou's lemma one gets that $Y \xi^\infty$ is also a $\prob$-supermartingale for all $Y \in \Y(1)$. Since there exists a local $\prob$-martingale in $\oY \in \Y(1)$ with $\oY_T > 0$, $\prob$-a.s., the optional decomposition theorem as appears in \cite{MR1804665} implies that $\xi^\infty \in \Xxx$.
\end{proof}

\subsection{NFLVR and the supermartingale property of wealth processes under a WELMM}

We now move to another characterization of the NFLVR condition using the concept of WELMMs. We start with a simple observation. If $\qprobb$ is a probability measure equivalent to $\prob$, it is straightforward to check that all $X \in \X$ are $\qprobb$-supermartingales if and only if $\inner{\qprobb}{X_T} \leq X_0$ for all $X \in \X$. Consider now an ELMM $\qprobb$. Since nonnegative local $\qprobb$-martingales are $\qprobb$-supermartingales, every $X \in \X$ is a $\qprobb$-supermartingale; therefore, $\inner{\qprobb}{X_T} \leq X_0$ for all $X \in \X$. One wonders, \emph{does the last property hold when $\qprobb$ is replaced by a \WELMM \ $\qprob$?}

Before we state and prove a result along the lines of the above discussion, some terminology will be introduced. A mapping $\qprob: \F \mapsto [0,1]$ will be called a \textsl{weakly equivalent finitely additive probability} if (1) and (2) of Definition \ref{defn: local prob, weak equiv} hold, as well as, $\prob$-a.s., $\ud \qprob^r / \ud \prob > 0$. Obviously, a local probability weakly equivalent to $\prob$ is a weakly equivalent finitely additive probability. A \textsl{separating weakly equivalent finitely additive probability} is a weakly equivalent finitely additive probability $\qprob$ such that
$\inner{\qprob}{X_T} \leq X_0$ for all $X \in \X$. We can then think of the processes $X \in \X$ as being $\qprob$-supermartingales. In accordance to the discussion above, the natural question that comes into mind is: when can we find a separating WELMM separating? In loose terms: \emph{can we find a \WELMM \ $\qprob$ such that all elements of $\X$ $\qprob$-supermartingales?} The answer, given in Theorem \ref{thm: WELMMS ans SEP-MEAS} below, is that this \emph{only} happens under the NFLVR condition.

\begin{thm} \label{thm: WELMMS ans SEP-MEAS}
The following are equivalent:
\begin{enumerate}
  \item The market satisfies the \NFLVR \ condition.
  \item There exists an \ELMM \ $\qprobb$.
  \item There exists a separating weakly equivalent finitely additive probability.
\end{enumerate}
\end{thm}

\begin{proof}
We prove $(1) \Rightarrow (3)$, $(3) \Rightarrow (2)$, and $(2) \Rightarrow (1)$ below.

\smallskip

\noindent $\mathbf{(1) \Rightarrow (2).}$ This is a consequence of \cite{MR1671792} and the fact that nonnegative $\sigma$-martingales are local martingales --- see \cite{MR1277002}.

\smallskip

\noindent $\mathbf{(2) \Rightarrow (3).}$ An ELMM is a separating weakly equivalent finitely additive probability.

\smallskip

\noindent $\mathbf{(3) \Rightarrow (1).}$ In view of Proposition 3.6 of \cite{MR1304434} and Proposition 1.3 proved previously in the present paper, condition NFLVR is equivalent to showing that (a) $\set{X_T \such X \in \X(1)}$ is bounded in probability, and (b) If $\prob[X_T \geq X_0] = 1$ for some $X \in \X$, then $\prob[X_T > X_0] = 0$. For (a), observe that
\[
\sup_{X \in \X(1)} \expec \bra{\pare{\frac{\ud \qprob^r}{\ud \prob}} X_T} = \sup_{X \in \X(1)} \inner{\qprob^r}{X_T} \leq \sup_{X \in \X(1)} \inner{\qprob}{X_T} \leq 1;
\]
in particular, $\set{(\ud \qprob^r / \ud \prob) X_T \such X \in \X(1)}$ is bounded in probability and, since $\prob[(\ud \qprob^r / \ud \prob) > 0] = 1$, $\set{X_T \such X \in \X(1)}$ is bounded in probability as well. To show (b), note that, for any $\epsilon > 0$ and $X \in \X$ with $\prob[X_T \geq X_0] = 1$, we have
\[
X_0 \geq \inner{\qprob}{X_T} \geq \inner{\qprob}{X_0 \indic_\Omega + \epsilon \indic_{\set{X_T > X_0 + \epsilon}}} = X_0 + \epsilon \qprob[X_T > X_0 + \epsilon] \geq X_0 + \epsilon \qprob^r[X_T > 1 + \epsilon].
\]
It follows that $\qprob^r[X_T > X_0 + \epsilon] = 0$; since $\prob[(\ud \qprob^r / \ud \prob) > 0] = 1$, this is equivalent to $\prob[X_T > X_0 + \epsilon] = 0$. The latter holds for all $\epsilon > 0$, so we get $\prob[X_T > X_0] = 0$, which completes the argument.
%
%
%
\end{proof}

\section{The Case of Continuous-Path Semimartingales} \label{sec: cont semimarts}

In this section, we shall state and prove a result that implies Theorem \ref{thm: main} in the case where $S$ is a $d$-dimensional continuous-path semimartingale. Note that Assumption \eqref{S-MART} will \emph{not} be in force here; in particular, there can be more than one traded security and the prices of securities do not have to be nonnegative. In fact, Theorem \ref{thm: FTAP_cont} that is presented below actually sharpens the conclusion of Theorem \ref{thm: main} by providing a further equivalence in terms of the local rates of return and local covariances of the discounted prices $S = (S^i)_{i=1, \ldots, d}$.

We first introduce some notation. Since $S$ is a continuous-path semimartingale, one has the decomposition $S= A + M$, where $A = (A^1, \ldots, A^d)$ has continuous paths and is of finite variation, and $M = (M^1, \ldots, M^d)$ is a continuous-path local martingale. Denote by $[M^i, M^k]$ the quadratic (co)variation of $M^i$ and $M^k$. Also, let $[M, M]$ be the $d \times d$ nonnegative-definite symmetric matrix-valued process whose $(i, k)$-component is $[M^i, M^k]$. Call now $G := \trace [M, M]$, where $\trace$ is the operator returning the trace of a matrix. Observe that $G$ is an increasing, adapted, continuous process and that there exists a $d \times d$ nonnegative-definite symmetric matrix-valued process $c$ such that $[M^i, M^k] = \int_0^\cdot c^{i,k}_t \ud G_t$; $[M, M] = \int_0^\cdot c_t \ud G_t$ in short.

\begin{thm} \label{thm: FTAP_cont}
In the above-described market, the following statements are equivalent:
\begin{enumerate}
    \item There are no arbitrages of the first kind in the market.
	 \item There exists a \emph{strictly positive} local $\prob$-martingale $Y$ with $Y_0 = 1$ such that $Y S^i$ is a local $\prob$-martingale for all $i \in \set{1, \ldots, d}$.
    \item There exists a $d$-dimensional, predictable process $\rho$ such that $A = \int_0^\cdot (c_t \rho_t) \ud G_t$, as well as $\int_0^T \inner{\rho_t}{c_t \rho_t} \ud G_t < \infty$.
\end{enumerate}
\end{thm}

\begin{proof}

We prove $(1) \Rightarrow (3)$, $(3) \Rightarrow (2)$, and $(2) \Rightarrow (1)$ below.

\smallskip

\noindent $\mathbf{(1) \Rightarrow (3).}$ We shall show that if statement (3) of Theorem \ref{thm: FTAP_cont} is not valid, then $\set{X_T \such X \in \X(1)}$ is not bounded in probability. In view of Proposition \ref{prop: NA_one iff NUPBR}, $(1) \Rightarrow (3)$ will be established.

Suppose that one \emph{cannot} find a predictable $d$-dimensional process $\rho$ such that $A = \int_0^\cdot (c_t \rho_t) \ud G_t$. In that case, linear algebra combined with a measurable selection argument gives the existence of some bounded predictable process $\theta$ such that (a) $\int_0^T \theta_t \ud G_t= 0$, (b) $\int_0^\cdot \inner{\theta_t}{\ud A_t }$ is a \emph{nondecreasing} process, and (c) $\prob[\int_0^T \inner{\theta_t}{\ud A_t}  > 0 ] > 0$. This of course means that $X^{1, \theta} \in \X(1)$, in the notation of \eqref{eq: wealth process, all}, satisfies $X^{1, \theta} \geq 1$, $\prob[X^{1, \theta}_T > 1] > 0$. Then, $X^{1, k \theta} \in \X(1)$ for all $k \in \Natural$ and $(X^{1, k \theta})_{k \in \Natural}$ is not bounded in probability.

Now, suppose that $A = \int_0^\cdot (c_t \rho_t) \ud G_t$ for some predictable $d$-dimensional process $\rho$, but that $\prob \left[ \int_0^T \inner{\rho_t}{c_t \rho_t} \ud G_t = \infty \right] > 0$. Consider the sequence $\pi^k := \rho \indic_{\set{|\rho| \leq k}}$ and let $X^k$ be defined via $X^k_0 = 1$ and satisfying $\ud X^k_t = X^k_t \pi^k_t \ud S_t$. Then, It\^o's formula implies that
\[
\log X^k_T = - \frac{E^k_T}{2} + \int_0^{T} \left( \rho_t \indic_{\set{|\rho_t| \leq k}} \right) \ud M_t,
\]
holds for all $k \in \Natural$, where $E^k_T := \int_0^T \inner{\rho_t}{c_t \rho_t} \indic_{\set{|\rho_t| \leq k}} \ud G_t$ coincides with the total quadratic variation of the local martingale $\int_0^{\cdot} \left( \rho_t \indic_{\set{|\rho_t| \leq k}} \right) \ud
M_t$. It follows that, for every $k \in \Natural$, one
can find a one-dimensional standard Brownian motion $\beta^k$ such that
\[
\log X^k_T = - \frac{E^k_T}{2} +
\beta^k_{E^k_T}.
\]
The strong law of large numbers for Brownian motion will imply that
\[
\lim_{k \to \infty}
\prob \left[ \left| \frac{\beta^k_{E^k_T}}{E^k_T} \right|
> \epsilon, \, \int_0^T \inner{\rho_t}{c_t \rho_t} \ud G_t = \infty \right] = 0, \text{ for all } \epsilon > 0,
\]
so that
\[
\lim_{k \to \infty} \prob \bra{\frac{\log X^k_T }{E^k_T} > \frac{1}{2}  -
\epsilon \, \Big| \, \int_0^T \inner{\rho_t}{c_t \rho_t} \ud G_t = \infty } = 1, \text{ for all } \epsilon > 0.
\]
Choosing $\epsilon = 1/4$, it follows that if $\prob \left[ \int_0^T \inner{\rho_t}{c_t \rho_t} \ud G_t = \infty \right] > 0$, the sequence $(X^k_T)_{k \in \Natural}$ is not bounded in probability.

\smallskip

\noindent $\mathbf{(3) \Rightarrow (2).}$ With the data of condition (3) there, define the process
\[
Y := \exp \left( - \int_0^\cdot \inner{\rho_t}{\ud S_t} + \frac{1}{2} \int_0^\cdot \inner{\rho_t}{c_t \rho_t} \ud G_t \right).
\]
Condition (3) ensures that $Y$ is well-defined (meaning that the two integrals above make sense). It\^o's formula easily shows that $Y$ is a local $\prob$-martingale. Then, a simple use of integration-by-parts gives that $Y S^i$ is a local martingale for all $i \in \set{1, \ldots, d}$.

\smallskip

\noindent $\mathbf{(2) \Rightarrow (1).}$ The proof of this implication is somewhat classic, but will be presented anyhow for completeness. Start with a sequence $(X^k)_{k \in \Natural}$ of wealth processes such that $\limk X_0^k = 0$ as well as $\prob$-$\limk X^k_T = \xi$ for some $[0, \infty]$-valued random variable $\xi$. Since $Y S^i$ is a local $\prob$-martingale for all $i \in \set{1, \ldots, d}$, a straightforward multidimensional generalization of the proof of Proposition \ref{prop: Z Q-mart iff Y Z P-mart} shows that, for all $k \in \Natural$, $Y X^k$ is a local $\prob$-martingale. As nonnegative local $\prob$-martingales are $\prob$-supermartingales, we have $\expec[Y_T X_T^k] \leq X_0^k$ holding for all $k \in \Natural$. Fatou's lemma implies now that $\expec[Y_T \xi] \leq \liminf_{k \to \infty} \expec[Y_T X_T^k] \leq \liminf_{k \to \infty} X_0^k = 0$.
Since $Y_T > 0$ and $\xi \geq 0$, $\prob$-a.s, the last inequality holds if only if $\prob[\xi = 0] = 1$. Therefore, $(X^k)_{k \in \Natural}$ is not an arbitrage of the first kind.
\end{proof}

\begin{rem}[Market price of risk and the num\'eraire portfolio]
Condition (3) of Theorem \ref{thm: FTAP_cont} has some economic consequences. Assume for simplicity that $G$ is absolutely continuous with respect to Lebesgue measure, i.e., that $G := \int_0^\cdot g_t \ud t$ for some predictable process $g$. Under condition NA$_1$, we also have $A := \int_0^\cdot a_t \ud t$ for some predictable process $g$, and that there exists a predictable process $\rho$ such that $c \rho = a$. (In fact, the latter process $\rho$ can be taken to be equal to $c^\dagger a$, where $c^\dagger$ is the Moore-Penrose pseudo-inverse of $c$.) Now, take $c^{1/2}$ to be any root of the nonnegative-definite matrix $c$ (that can be chosen in a predictable way) and define $\sigma := c^{1/2} \sqrt{g}$. Then, we can write $\ud S_t = \sigma_t \big( \lambda_t \ud t + \ud W_t \big)$, where $W$ is a standard $d$-dimensional Brownian motion\footnote{In the case where $c$ is nonsingular for Lebesgue-almost every $t \in \Real$, $\prob$-almost surely, we have $W_\cdot := \int_0^\cdot c_t^{-1/2} \ud M_t$. If $c$ fails to be nonsingular for Lebesgue-almost every $t \in \Real$, $\prob$-almost surely, one can still construct a Brownian motion $W$ in order to have $M_\cdot = \int_0^\cdot c_t^{1/2} \ud W_t$ holding by enlarging the probability space --- check for example \cite{MR1121940}, Theorem 4.2 of Section 3.4.} and $\lambda \dfn \sigma^\top \rho$ is a \textsl{risk premium} process (in the one-dimensional case also commonly known as the \textsl{Sharpe ratio}), that has to satisfy $\int_0^T |\lambda_t|^2 \ud t < \infty$ for all $T \in \Real_+$. We conclude that condition NA$_1$ is valid if and only if a risk-premium process exists and is locally square-integrable in a pathwise sense.
\end{rem}


\bibliographystyle{siam}
\bibliography{na1}
\end{document}